\documentclass[12pt,a4paper]{article}

\usepackage{amsmath,amsthm,amssymb,amscd,a4wide}

\usepackage{dsfont}
\usepackage{mathrsfs}
\usepackage{setspace}

\newcommand{\ud}{\mathrm{d}}

\newcommand{\ue}{\mathrm{e}}

\newcommand{\cD}{{\mathcal D}}

\newcommand{\cH}{{\mathcal H}}

\newcommand{\rz}{{\mathbb R}}

\newcommand{\kz}{{\mathbb C}}

\DeclareMathOperator{\re}{Re}
\DeclareMathOperator{\tr}{Tr}

\newcommand{\eins}{\mathds{1}}

\numberwithin{equation}{section}

\newtheorem{theorem}{Theorem}[section]
\newtheorem{lemma}[theorem]{Lemma}

\theoremstyle{definition}
\newtheorem{defn}[theorem]{Definition}

\numberwithin{equation}{section}

\begin{document} 

\thispagestyle{empty}

\vspace*{1cm}

\begin{center}
  
{\LARGE\bf Instability of Bose-Einstein condensation\\[5mm]
on quantum graphs\\[5mm] 
under repulsive perturbations} \\

\vspace*{2cm}

{\large Jens Bolte}%
\footnote{E-mail address: {\tt jens.bolte@rhul.ac.uk}}

\vspace*{5mm}

Department of Mathematics\\
Royal Holloway, University of London\\
Egham, TW20 0EX\\
United Kingdom\\

\vspace*{1cm}

{\large Joachim Kerner}%
\footnote{E-mail address: {\tt Joachim.Kerner@fernuni-hagen.de}}

\vspace*{5mm}

Abteilung Analysis\\
FernUniversit\"{a}t in Hagen\\
58084 Hagen\\
Germany\\

\end{center}

\vfill
 
\begin{abstract} 
In this Note we investigate Bose-Einstein condensation in interacting quantum
many-particle systems on graphs. We extend previous results obtained for 
particles on an interval and show that even arbitrarily small repulsive 
two-particle interactions destroy a condensate in the non-interacting Bose gas. 
Our results also cover singular two-particle interactions, such as the 
well-known Lieb-Lininger model, in the thermodynamic limit.
\end{abstract}

\noindent
{\bf\small MSC classification:} 82D50, 81V70, 46N50\\
\noindent
{\bf\small Keywords:} Bose-Einstein condensation, quantum graphs, Lieb-Liniger 
model
\newpage

\section{Introduction}
Since its prediction by Einstein \cite{EinsteinBECI,EinsteinBECII}, Bose-Einstein
condensation (BEC) has played an important role in many-particle quantum 
mechanics. From a physical point of view, the condensation refers to the fact 
that an infinite number of particles occupy the same one-particle state and that
this macroscopic occupation of the same state induces a coherent behaviour in 
the many-particle system bearing some resemblance with a superconducting 
behaviour. Indeed, this observation led to the well-known Bogoliubov mean-field 
model of BEC (see, e.g., \cite{ZagBru}). On the other hand, from a mathematical 
point of view the difficulties to establish BEC in interacting systems were soon
realised. This applies, in particular, to continuous systems and only recently 
rigorous results have been obtained for them 
\cite{BSP83,LiebSeiringerProof,LVZ03}. 

A related question to proving BEC in an interacting system is the following: 
Suppose that a non-interacting many-particle system shows BEC as, e.g., a free 
gas in a three-dimensional box. The question then arises, whether this 
condensation is stable under perturbing the free gas by introducing repulsive 
particle interactions. In a number of cases it had been shown that hardcore 
interactions destroy BEC \cite{BP85,BP86,PuleAonghusa87}, and this was recently 
confirmed for many-particle quantum systems on (finite, metric) graphs 
\cite{BolteKernerBEC}. However, since hardcore interactions can be viewed 
as very strong one might wonder whether small repulsive interactions can be 
implemented such that the condensation survives. It appears that the answer 
strongly depends on the type of condensation in the free system. Whereas 
the results of \cite{LVZ03} show that under some circumstances BEC is stable 
with respect to superstable repulsive two-particle interactions, it was shown in
another example \cite{PS86} that even small repulsive interactions destroy the 
condensate. The reason for this very different behaviour lies in the nature of 
the one-particle ground state of the free gas. Whereas in the model of 
\cite{LVZ03} this ground state is a plane wave and, hence, completely 
delocalised, the ground state in the model of \cite{PS86} is localised at the 
boundary of the system. Intuitively, this explains why the effect of repulsive 
interactions is much stronger in the latter system, leading to the destruction 
of the condensate.

It is our goal in this Note to prove that BEC of a free gas on a quantum
graph is unstable under the addition of a fairly general repulsive two-particle 
interaction. Quantum graphs are (ramified) one-dimensional models with a 
potentially complex topology. Although quantum graphs are studied in various 
areas of mathematics and physics (see, e.g., \cite{BerKuc13}), they are 
particularly prominent in the field of quantum chaos \cite{KS97}. This is due 
to the fact that the spectral correlations of sufficiently complex quantum 
graphs follow the Bohigas-Gianonni-Schmit conjecture \cite{BohGiaSch84} and are 
well described by random matrix theory. Contrary to what is often believed, BEC 
can occur in a free gas in one dimension when attractive boundary conditions are
chosen \cite{LanWil79}. In a similar spirit, a free Bose gas on a finite graph 
can show BEC at finite temperature. We previously identified the class of 
boundary conditions in the vertices of the graph that lead to BEC of a free gas 
and showed that any condensation (in terms of singularities of the free energy) 
is destroyed by adding hardcore two-particle interactions \cite{BolteKernerBEC}.
In this Note we now allow repulsive two-particle interactions to be generated 
by a potential in the same way as in \cite{PS86}, where the case of an interval 
(a graph with one edge in our language) with attractive boundary conditions at 
one end was investigated. However, here we consider arbitrary (finite) graphs 
and also include interactions approaching the Lieb-Liniger model \cite{LL63} 
(constructed on graphs in \cite{BKContact}) in the thermodynamic limit. In all 
cases we prove that BEC into the one-particle ground state is destroyed by the 
interactions at any finite temperature.

It is important to note that besides being one of the few explicitly solvable 
many-particle models, rigorous results concerning BEC in the Lieb-Liniger model 
are scarce. Recently, e.g., condensation in the Lieb-Liniger model on an 
interval with additional random potentials was discussed in \cite{SeiYngZag12}, 
proving condensation at zero temperature in a Gross-Pitaevskii regime. The 
thermodynamic limit employed in \cite{SeiYngZag12} is a high-density limit 
since the volume of the one-particle configuration space is not changed. In 
this Note, however, we study BEC on general (finite) graphs at finite 
temperature in the standard thermodynamic limit with fixed density.

\section{Background}
\label{sec1}
In this section we briefly summarise relevant concepts of one-particle and 
many-particle quantum graphs, as well as Bose-Einstein condensation. For more 
details on BEC see~\cite{PO56,Lieetal05,VerbeureBook}, on quantum graphs 
see~\cite{Kuc04,KS06,GnuSmi07,BolEnd09} and on many-particle quantum graphs 
see~\cite{BKSingular,BKContact}. For a discussion of BEC on quantum graphs 
see~\cite{BolteKernerBEC}.

Let $\Gamma=(\mathcal{V},\mathcal{E})$ be a finite graph with vertices 
$\mathcal{V}=\{v_1,\dots,v_V\}$ and edges $\mathcal{E}=\{e_1,\dots,e_E\}$ 
connecting the vertices. The graph is equipped with a metric structure by
assigning a (finite) length $l_e>0$ to each edge $e \in \mathcal{E}$. Hence,
each edge $e$ is associated with an interval $[0,l_e]$, and this allows us to 
define the one-particle Hilbert space, 
\begin{equation}
\label{QGHilbert1}
 \cH_1=\bigoplus_{e\in\mathcal{E}}L^2(0,l_e)\ .
\end{equation}
In order to obtain a quantum graph one introduces a self-adjoint realisation of 
the Laplacian in $\cH_1$. As a differential expression the Laplacian acts on 
smooth functions $F=(f_1,\dots,f_E) \in \cH_1$ as
\begin{equation}
\label{LaplDiff}
 -\Delta_1 F=(-f^{\prime \prime}_1,\dots,-f^{\prime \prime}_E)\ .
\end{equation}
Here the index $1$ refers to the fact that \eqref{LaplDiff} is a one-particle 
operator which, in our case, serves as the one-particle Hamiltonian. 

There are several ways to characterise self-adjoint realisations of the 
differential expression \eqref{LaplDiff}, see \cite{KosSch99,Kuc04}. In the 
following we shall refer to the approach developed in \cite{Kuc04}. This 
characterises the domains $\cD_1(P_1,L_1)\subset\cH_1$ on which $-\Delta_1$ 
is self-adjoint in terms of two linear maps $P_1,L_1$ on $\kz^{2E}$, where 
$P_1$ is a projection and $L_1$ is a self-adjoint endomorphism on $\ker P_1$.
These maps act on the boundary values of functions and their derivatives on
the edges and hence implement the connectivity of the graph. Any self-adjoint 
realisation of the one-particle Laplacian has compact resolvent. Its spectrum, 
therefore, is purely discrete, with an eigenvalue count following a Weyl 
asymptotics. Moreover, there are at most finitely many negative eigenvalues, 
whose number is bounded by the number of positive eigenvalues of $L_1$ 
\cite{KS06}. 

Following the usual construction, the bosonic $N$-particle Hilbert space is 
the symmetrised $N$-fold tensor product of the one-particle Hilbert space
\eqref{QGHilbert1}, i.e., $\cH_B^N=\cH_1 \otimes_s\dots\otimes_s\cH_1$. 
Accordingly, the $N$-particle Hamiltonian is given by 
\begin{equation}
\label{NLaplace}
 -\Delta_N=\sum_{j=1}^{N}\eins \otimes \dots \otimes (-\Delta_1) \otimes \dots 
 \otimes \eins \ .
\end{equation}
As a differential expression this is a Laplacian in $N$ variables. A number of 
self-adjoint realisations of $-\Delta_N$ are discussed in 
\cite{BKSingular,BKContact}, including non-interacting as well as interacting
ones. Non-interacting realisations $(-\Delta_N,\cD_N(P_1,L_1))$ follow from a 
tensor product construction, where the domain $\cD_N(P_1,L_1)$ is an $N$-fold 
tensor power of the one-particle domain $\cD_1(P_1,L_1)$. An important 
consequence is that the spectrum of $(-\Delta_N,\cD_N(P_1,L_1))$ consists of
eigenvalues that are sums of one-particle eigenvalues. For more details see
\cite{BolteKernerBEC}. Among the interacting realisations of the $N$-particle
Laplacian developed in \cite{BKContact} is a rigorous version of
\begin{equation}
\label{LL}
 H_N = -\Delta_N + \alpha\sum_{i<j}\delta(x^{(i)} -x^{(j)})\ ,
\end{equation}
defining a Lieb-Liniger model (see \cite{LL63}) on a graph. Here $i,j$ are
particle labels attached to coordinates on the same edge. In this 
context $-\Delta_N$ stands for a non-interacting realisation of the Laplacian; 
the interaction is a singular two-particle contact interaction, as indicated 
by the $\delta$-potentials.

In bosonic many-particle systems, BEC refers to the macroscopic occupation of a 
one-particle state. We work in the canonical ensemble, i.e., with a fixed 
particle number $N$ and Hilbert space $\cH_B^N$. The expectation value of a 
(bounded) observable $A_N$ in the Gibbs state $\omega_\beta$ at inverse 
temperature $\beta$ therefore is
\begin{equation}
\label{Gibbs}
 \omega_{\beta}(A_N) = \frac{1}{Z_N(\beta)}\tr(A_N e^{-\beta H_N})\ ,
\end{equation}
where $Z_N(\beta)=\tr(e^{-\beta H_N})$ is the partition function and $H_N$ is the 
$N$-particle Hamiltonian operator that is assumed to preserve the particle 
number. Condensation only takes place in the thermodynamic limit, which in the 
canonical ensemble is obtained by letting the volume of the one-particle 
configuration space tend to infinity while keeping the particle density constant 
\cite{GallavottiBook}. In the context of quantum graphs, the volume is the total 
length $\mathcal{L}=\sum_{e=1}^{E}l_e$ of the graph.  
\begin{defn} 
In a quantum graph the thermodynamic limit is obtained by rescaling each edge 
length $l_e$ as $nl_e$, and letting $n \to \infty$. At the same time 
the number of particles $N$ is increased such that the particle density 
$\rho=N/\mathcal{L}$ remains constant.
\end{defn}
We denote the thermodynamic limit by writing $\mathcal{L} \to \infty$.

Although, strictly speaking, not needed in the canonical ensemble, some tools 
of second quantisation will be useful in the following. Hence, we let 
$\mathcal{F}_B$ be the bosonic Fock space over the one-particle Hilbert space 
$\cH_1$ defined in \eqref{QGHilbert1}. When 
$\Phi=(\varphi_e)_{e\in\mathcal{E}}\in \cH_1$, the standard annihilation and 
creation operators in $\mathcal{F}_B$ are $a(\Phi)$ and $a^\ast(\Phi)$. They can 
be represented as
\begin{equation}
\begin{split}
\label{aadagger}
 a(\Phi)  &= \sum_{e=1}^{E}\int_{0}^{l_e}\overline{\varphi_e}(x)a_e(x)\ \ud x\ , \\
 a^\ast(\Phi)  &=\sum_{e=1}^{E}\int_{0}^{l_e}\varphi_e(x)a^\ast_e(x)\ \ud x\ ,
\end{split}
\end{equation} 
with operator-valued distributions $a_e$ and $a^\ast_e$ attached to each edge
that satisfy the CCR
\begin{equation}
\label{CCR}
 [a_e(x),a_{e'}(y)]=0=[a^\ast_e(x),a^\ast_{e'}(y)]\ ,\qquad
 [a_e(x),a^\ast_{e'}(y)]=\delta_{ee'}\delta(x-y)\ .
\end{equation} 
The particle number operator $\hat{N}$ can be expressed as
\begin{equation}
\hat{N}=\sum_{e=1}^{E}\int_{0}^{l_e} a^\ast_e(x)a_e(x)\ \ud x\ .
\end{equation}
If the $N$-particle system is in the Gibbs state \eqref{Gibbs}, the number of
particles in a given one-particle state $\Phi\in\cH_1$ is 
$\omega_{\beta}(a^*(\Phi)a(\Phi))$.
\begin{defn} 
Let $\Phi\in\cH_1$ be any (normalised, pure) one-particle state. We say 
that Bose-Einstein condensation is exhibited in the state $\Phi$ at 
inverse temperature $\beta > 0$, if
\begin{equation}\label{ParticleDensity}
 \limsup_{\mathcal{L}\to\infty}\frac{\omega_{\beta}(a^*(\Phi)a(\Phi))}{\mathcal{L}} 
 > 0\ .
\end{equation}
\end{defn}
In \cite{BolteKernerBEC} non-interacting Bose gases on quantum graphs were
classified according to whether or not they show BEC: 
\begin{theorem}
\label{TheoremNonInteracting} 
Let $(-\Delta_1,\cD_1(P_1,L_1))$ be a one-particle Laplacian on a graph and
denote the associated non-interacting $N$-particle Laplacian by 
$(-\Delta_N,\cD_N(P_1,L_1))$. Then BEC occurs for $(-\Delta_N,\cD_N(P_1,L_1))$
below some critical temperature, if and only if the map $L_1$ has at least one 
positive eigenvalue.
\end{theorem}
The key mechanism that leads to condensation in the non-interacting Bose gas is 
a gap in the one-particle spectrum that separates a finite number of states, 
in particular the ground state, from the states of positive energy. It is known
that there can only be negative eigenvalues of the one-particle Laplacian
when $L_1$ is not negative semi-definite, hence the requirement for $L_1$ to
possess a positive eigenvalue. Furthermore, according to Lemma~3.3 in 
\cite{BolteKernerBEC}, the ground state eigenvalue converges to $-L^2_{max}$ in 
the thermodynamic limit, where $L_{max}$ is the largest eigenvalue of $L_1$. 
Therefore, the essential condition for BEC is a gap in the one-particle spectrum 
that persists in the thermodynamic limit.

It is generally assumed that a gap in the one-particle spectrum stabilises 
condensation, or even makes it possible at all \cite{LVZ03}. This is true, in 
particular, for a one-dimensional Bose gas since no condensation is present for 
standard boundary conditions as expressed by the well-known result of 
Hohenberg~\cite{Hohenberg}. In three dimensions, where condensation occurs 
without a gap, an additional gap stabilises the condensate. As an example, in 
\cite{LVZ03} an artificial gap was introduced, and it was shown that a certain 
class of repulsive two-particle interactions does not destroy the condensate. 
However, although a gap in the one-particle spectrum generally stabilises 
condensation, it still may not be possible to implement repulsive two-particle 
interactions without destroying the condensate. An example for this was studied 
in \cite{PS86}. The reason for the differences in the examples of \cite{LVZ03} 
and \cite{PS86} lies in the strongly localised nature of the ground state of 
the model studied in \cite{PS86}. In the cases covered by 
Theorem~\ref{TheoremNonInteracting} where BEC occurs, the ground states are
localised around the vertices of the graph. One thus expects a similar behaviour
to the one found in \cite{PS86}.

\section{Results}
In this section we start with non-interacting Bose gases that, according to 
Theorem~\ref{TheoremNonInteracting}, show BEC and then investigate the effect of 
additional, repulsive two-particle interactions. Working in the canonical 
ensemble we shall define a restriction of the $N$-particle Hamiltonian to the 
(finite) one-particle configuration space \cite{GallavottiBook} and then 
investigate the particle number density \eqref{ParticleDensity} of the 
one-particle ground state in the limit $\mathcal{L} \rightarrow \infty$. More 
explicitly, the $N$-particle Hamiltonian is given by 
\begin{equation}
 H_N=-\Delta_N + U^{\mathcal{L}}_{N}\ ,
\end{equation} 
where $(-\Delta_N,\cD_N(P_1,L_1))$ is such that $L_1$ has at least one positive
eigenvalue. Therefore, according to Theorem~\ref{TheoremNonInteracting}, the 
free Bose gas with Hamiltonian $-\Delta_N$ shows BEC below a critical 
temperature. The interaction potential $U^{\mathcal{L}}_{N}$ is defined in terms of 
a function $U_{\mathcal{L}}:\rz \rightarrow \rz_+$ such that, in the language of 
second quantisation,
\begin{equation}\label{SecondQuantisedHamiltonian}
 U_N^{\mathcal{L}} = \frac{1}{2}\sum_e\int_{0}^{l_e} \int_{0}^{l_e}a_e^\ast(x)a_e^\ast(y)
 U_{\mathcal{L}}(x-y)a_e(x)a_e(y)\ \ud x\,\ud y\ ,
\end{equation} 
see also \eqref{aadagger} and \eqref{CCR}. We require the function 
$U_{\mathcal{L}}$ to be (i) positive, in order to generate repulsive interactions, 
and (ii) to result in a Kato-Rellich perturbation $U_N^{\mathcal{L}}$ of 
$-\Delta_N$, so that $H_N$ is self-adjoint on the domain $\cD_N(P_1,L_1)$.
Furthermore, we assume that
\begin{enumerate}
\item $\|U_{\mathcal{L}}\|_{L^1(\rz)}$ is finite and independent of $\mathcal{L}$,
\item For all $\mathcal{L}>0$ there exists $A_\mathcal{L},\epsilon_\mathcal{L}>0$ 
such that $U_\mathcal{L}(x) \geq \epsilon_\mathcal{L}$ for all 
$x \in [-A_\mathcal{L},+A_\mathcal{L}]$.
\end{enumerate}
More specifically, we require that $A_\mathcal{L}$ is either independent of 
$\mathcal{L}$, or that $\lim_{\mathcal{L}\to \infty}A_\mathcal{L}=0$. These families 
of potentials include $\delta$-sequences, e.g., of the form 
$U_\mathcal{L}(x):=\mathcal{L}V(\mathcal{L}x)$ with $V\in C_0^\infty(\rz)$, 
$V\geq 0$ and $\|V\|_{L^1(\rz)}=\alpha>0$, so that 
$\lim_{\mathcal{L}\to\infty}U_\mathcal{L}(x)=\alpha\delta(x)$. This implies that in 
the thermodynamic limit $\mathcal{L}\to\infty$ we may include Lieb-Liniger 
models, see \eqref{LL} and \cite{BKContact}.

The following result is adapted from \cite{PS86}
\begin{lemma}
\label{LemmaEnergy} 
Let $U_\mathcal{L}$ be a sequence of potentials with the properties described 
above. Then the energy density remains finite in the thermodynamic limit, 
i.e.,
\begin{equation}
\label{EquationEnergy}
 \limsup_{\mathcal{L}\to\infty}\frac{\omega_{\beta}(H_N)}{\mathcal{L}} < \infty\ .
\end{equation} 
\end{lemma}
\begin{proof}
The proof follows the strategy outlined in \cite{PS86}. It uses a normalised 
one-particle vector $\Phi\in\cH_1$ such that each component 
$(\Phi)_e=\varphi_e\in H^1(0,l_e)$ is supported in $(0,l_e)$, bounded in 
absolute value by $\frac{1}{\sqrt{El_e}}$, and equal to 
$\frac{d_e}{\sqrt{El_e}}$ on the interval $[a,l_e-a]$ for some $a>0$ with 
$d_e \rightarrow 1$ as $l_e \rightarrow \infty$.

Furthermore, we require $\Phi$ to be such that there is a constant $c_1 > 0$ 
with 
\begin{equation}
 \|\nabla\varphi_e\|^2_{L^2(0,l_e)} \leq c_1\ ,\quad \forall e \in \mathcal{E}\ .
\end{equation} 
Due to the repulsive nature of the potential one has
\begin{equation}
 f_{\mathcal{L}}(\beta)=\frac{1}{\beta \mathcal{L}}\log \tr(e^{-\beta H_N})
 \leq f^0_{\mathcal{L}}(\beta)=\frac{1}{\beta \mathcal{L}}\log\tr(e^{\beta \Delta_N})\ .
\end{equation} 
Defining $\Psi_N=\Phi\otimes\dots\otimes\Phi$ one gets 
\begin{equation}
\begin{split}
 \tr(e^{-\beta H_N}) &\geq e^{-\beta \langle \Psi_N,H_N\Psi_N\rangle_{\cH_N}} \\
                  &\geq e^{-\beta(NEc_1+c_2\frac{N(N-1)}{2\mathcal{L}}\|U_\mathcal{L}\|_{L^1})}\ ,
\end{split}
\end{equation} 
where $c_2> 0$ is a constant. As a consequence,
\begin{equation}
\label{ReferenceProof1a}
 -\Bigl(Ec_1\rho+c_2\frac{\rho^2}{2}\Bigr)-\epsilon \leq 
 f_{\mathcal{L}}(\beta) \leq f^0_{\mathcal{L}}(\beta)
\end{equation} 
for $\mathcal{L}$ large enough, with some $\epsilon > 0$. Finally, using 
the convexity of $f_{\mathcal{L}}(\beta)$,
\begin{equation}
 -\frac{\omega_{\beta}(H_N)}{\mathcal{L}} = 
  \frac{\ud f_{\mathcal{L}}}{\ud \beta}(\beta) 
  \geq \frac{f_{\mathcal{L}}(\beta)-f_{\mathcal{L}}(\beta-\delta)}{\delta}\ ,
\end{equation} 
the Lemma then follows using the bounds \eqref{ReferenceProof1a} and taking 
into account that \linebreak$\limsup_{\mathcal{L}\to\infty}f^0_{\mathcal{L}}(\beta)$ 
exists. The latter property follows from a bracketing argument and the explicit
knowledge of the eigenvalues for Dirichlet and Neumann vertex conditions 
(see, e.g., \cite{BKSingular}). 
\end{proof}
Following \cite{PS86}, the general idea is to show that BEC into the 
one-particle ground state, after repulsive interactions are switched on, would 
contradict Lemma~\ref{LemmaEnergy}. 
\begin{lemma}
\label{LemmaCutoff} 
Let $\Phi=(\varphi_e)_{e\in\mathcal{E}} \in \cH_1$ be a pure one-particle state. 
Define $\Phi_1$ and $\Phi_2$ as
\begin{equation}
\begin{split}
 (\Phi_1)_e&:=\varphi_e\chi_{[0,l^{\delta}_{\min}]}, \\
 (\Phi_2)_e&:=\varphi_e\chi_{[l_e-l^{\delta}_{\min},l_e]}, 
\end{split}
\end{equation} 
where $\delta < \frac{1}{3}$ is some constant, $l_{\min}$ is the shortest edge
length and $\chi_I$ is the characteristic function of the interval $I$. Then, 
given that the potential $U_{\mathcal{L}}$ described above is such that 
$\epsilon_{\mathcal{L}}$ and $A_{\mathcal{L}}$ are both constant or 
$\epsilon_{\mathcal{L}}A_{\mathcal{L}}^3=O(\mathcal{L}^{3\delta+\gamma-1})$ with 
$\gamma < 1-3\delta$, one has
\begin{equation}
\label{partdense}
 \limsup_{\mathcal{L} \to \infty}
 \frac{\omega_{\beta}(a^\ast(\Phi_j)a(\Phi_j))}{\mathcal{L}}=0\ , \quad j=1,2\ .
\end{equation}
\end{lemma}
\begin{proof} 
We follow the strategy outlined in \cite{PS86}. For this, we partition the interval 
$[0,l^{\delta}_{\min}]$ into $\lceil l^{\delta}_{\min}/A_{\mathcal{L}}\rceil$ sub-intervals 
$I_j=[(j-1)A_{\mathcal{L}},jA_{\mathcal{L}}]$, 
$0 < j < \lceil l^{\delta}_{\min}/A_{\mathcal{L}}\rceil$. In the same way, we partition 
$[l_e-l^{\delta}_{\min},l_e]$ into $\lceil l^{\delta}_{\min}/A_{\mathcal{L}}\rceil$ 
sub-intervals $\tilde{I}_j=[l_e-jA_{\mathcal{L}},l_e-(j-1)A_{\mathcal{L}}]$, 
$0 < j < \lceil l^{\delta}_{\min}/A_{\mathcal{L}}\rceil$. We then estimate:
\begin{equation}
\label{1stest}
\begin{split}
 \frac{\omega_{\beta}(H_N)}{\mathcal{L}} 
  &\geq -|E_0|\frac{N}{\mathcal{L}}\\
  &\quad +\frac{1}{2\mathcal{L}}\sum_{e}\int_{0}^{l_e} \int_{0}^{l_e}
     U_{\mathcal{L}}(x-y)\omega_{\beta}(a_e^\ast(x)a_e^\ast(y)a_e(x)a_e(y))\ 
     \ud x\,\ud y \\
  &\geq -|E_0|\frac{N}{\mathcal{L}}\\
  &\quad +\frac{1}{2\mathcal{L}}\sum_{e}\sum_{j=1}^{\lceil l^{\delta}_{\min}/A_{\mathcal{L}}\rceil}
     \int_{I_j}\int_{I_j} U_{\mathcal{L}}(x-y)
     \omega_{\beta}(a_e^\ast(x)a_e^\ast(y)a_e(x)a_e(y))\ \ud x\,\ud y \\
  &\quad +\frac{1}{2\mathcal{L}}\sum_{e}
     \sum_{j=1}^{\lceil l^{\delta}_{\min}/A_{\mathcal{L}}\rceil}\int_{\tilde{I}_j}\int_{\tilde{I}_j}
     U_{\mathcal{L}}(x-y) \omega_{\beta}(a_e^\ast(x)a_e^\ast(y)a_e(x)a_e(y))
     \ \ud x\,\ud y\ ,
\end{split}
\end{equation}
where $E_0<0$ is the ground-state eigenvalue of the one-particle Laplacian 
$-\Delta_1$. Now, using the lower bound on the potential, we obtain
\begin{equation}
\begin{split}
 \frac{\omega_{\beta}(H_N)}{\mathcal{L}} 
   &\geq -|E_0|\frac{N}{\mathcal{L}}\\
   &\quad +\frac{\epsilon_{\mathcal{L}}}{2\mathcal{L}}\sum_{e}
      \sum_{j=1}^{\lceil l^{\delta}_{\min}/A_{\mathcal{L}}\rceil}\int_{I_j}\int_{I_j}
      \omega_{\beta}(a_e^\ast(x)a_e^\ast(y)a_e(x)a_e(y))\ \ud x\,\ud y \\
   &\quad +\frac{\epsilon_{\mathcal{L}}}{2\mathcal{L}}\sum_{e}
      \sum_{j=1}^{\lceil l^{\delta}_{\min}/A_{\mathcal{L}}\rceil}\int_{\tilde{I}_j}\int_{\tilde{I}_j}
      \omega_{\beta}(a_e^\ast(x)a_e^\ast(y)a_e(x)a_e(y))\ \ud x\,\ud y\ .
\end{split}
\end{equation}
We define $\varphi^{(i)}_e:=\varphi_e\chi_{I_{i}}$ and 
$\tilde{\varphi}^{(i)}_e:=\varphi_e\chi_{\tilde{I}_i}$ 
as the components of functions $\Phi_{ji},\tilde{\Phi}_{ji}\in\cH_1$, such that
$(\Phi_{ji})_e=\delta_{ej}\varphi^{(i)}_j$ and 
$(\tilde{\Phi}_{ji})_e=\delta_{ej}\tilde{\varphi}^{(i)}_j$. However, for simplicity
we restrict our attention in the following to $\Phi_{ji}$. Using the 
Cauchy-Schwarz inequality for the Gibbs state \cite{VerbeureBook} we then obtain,
\begin{equation}
\label{Estimate1}
\begin{split}
 |\omega_{\beta}(a^\ast(\Phi_{ji})a(\Phi_{lk}))|^4 
   &\leq \omega_{\beta}^2(a^\ast(\Phi_{ji})a(\Phi_{ji})) \omega_{\beta}^2
       (a^\ast(\Phi_{lk})a(\Phi_{lk})) \\
   &\leq \omega_{\beta}(a^\ast(\Phi_{ji})a(\Phi_{ji})a^\ast(\Phi_{ji})a(\Phi_{ji})) \\
   &\qquad \cdot \omega_{\beta}(a^\ast(\Phi_{lk})a(\Phi_{lk})a^\ast(\Phi_{lk})
       a(\Phi_{lk}))\ ,
\end{split}
\end{equation}
and 
\begin{equation}
\label{Estimate2}
\begin{split}
 \omega_{\beta}(a^\ast(\Phi_{ji})a(\Phi_{ji})a^\ast(\Phi_{ji})a(\Phi_{ji}))
   &=\omega_{\beta}(a^\ast(\Phi_{ji})a^\ast(\Phi_{ji})a(\Phi_{ji})a(\Phi_{ji}))\\
   &\qquad +\|\varphi^{(i)}_j\|^2_{L^2(0,l_j)}\ 
            \omega_{\beta}(a^\ast(\Phi_{ji})a(\Phi_{ji}))\ .
\end{split}
\end{equation}
Next we establish two useful estimates. First, using the H\"{o}lder and then 
again the Cauchy-Schwarz inequality, yields,
\begin{equation}
\label{Estimate3}
\begin{split}
 \omega_{\beta}(a^\ast(\Phi_{ji})a(\Phi_{ji}))
   &=\int_{0}^{l_j}\int_{0}^{l_j}\varphi^{(i)}_{j}(x)\bar{\varphi}^{(i)}_{j}(y)
      \omega_{\beta}(a_j^\ast(x)a_j(y))\ \ud x\,\ud y \\
   &\leq\int_{0}^{l_j}\omega_{\beta}(a_j^\ast(x)a_j(x))\ \ud x \int_{0}^{l_j}
      |\varphi^{(i)}_{j}(y)|^2\ \ud y \\
   &\leq N \int_{0}^{l_j}|\varphi^{(i)}_{j}(y)|^2\ \ud y \ .
\end{split}
\end{equation}
Second, again using the H\"{o}lder and the Cauchy-Schwarz inequality, 
\begin{equation}
\label{Estimate4}
\begin{split}
 \omega_{\beta}(a^\ast(\Phi_{ji})a^\ast(\Phi_{ji})a(\Phi_{ji})a(\Phi_{ji}))
   &\leq \int_{I_i}\int_{I_i}
      \omega_{\beta}(a_j^\ast(x)a_j^\ast(y)a_j(x)a_j(y))\ \ud y\,\ud x \\
   &\qquad\qquad\cdot \left(\int_{I_i}\ 
      |\varphi^{(i)}_{j}(x)|^2\right)^2\ \ud x\\
   &\leq \int_{I_i}\int_{I_i}
      \omega_{\beta}(a_j^\ast(x)a_j^\ast(y)a_j(x)a_j(y))\ \ud y\,\ud x \\
   &:= C_{ji}\ .
\end{split}
\end{equation}
Combining \eqref{Estimate1}, \eqref{Estimate2}, \eqref{Estimate3} and 
\eqref{Estimate4} we obtain 
\begin{equation}
 \sum_{i,k=1}^{\lceil l^{\delta}_{\min}/A_{\mathcal{L}}\rceil}|\omega_{\beta}(a^\ast(\Phi_{ji})
 a(\Phi_{lk}))|^4 
 \leq \left(\sum_{i=1}^{\lceil l^{\delta}_{\min}/A_{\mathcal{L}}\rceil} C_{ji}+N\right)\left(
 \sum_{k=1}^{\lceil l^{\delta}_{\min}/A_{\mathcal{L}}\rceil} C_{lk}+N\right) \ .
\end{equation}
Using the inequality $|\sum_{j=1}^{n} a_{j}|^4 \leq n^3 \sum_{j=1}^{n}|a_{j}|^4$ then
gives 
\begin{equation}
\begin{split}
 |\omega_{\beta}(a^\ast(\Phi_{1})a(\Phi_{1}))|^4
   &=\left|\sum_{j,l=1}^{E}\sum_{i,k=1}^{\lceil l^{\delta}_{\min}/A_{\mathcal{L}}\rceil }\omega_{\beta}
      (a^\ast(\Phi_{ji})a(\Phi_{lk}))\right|^4 \\
   &\leq 2\frac{E^6l^{6\delta}_{\min}}{A_{\mathcal{L}}^6}\sum_{j,l=1}^{E}
      \sum_{i,k=1}^{\lceil l^{\delta}_{\min}/A_{\mathcal{L}}\rceil}
      |\omega_{\beta}(a^\ast(\Phi_{ji})a(\Phi_{lk})|^4\\
   &\leq 2\frac{E^6l^{6\delta}_{\min}}{A_{\mathcal{L}}^6}\left(\sum_{j=1}^{E}
      \sum_{i=1}^{\lceil l^{\delta}_{\min}/A_{\mathcal{L}}\rceil} C_{ji}+EN\right)^2\ .
\end{split}
\end{equation}
Hence, 
\begin{equation}
\label{Csum}
\begin{split}
 \frac{\epsilon_{\mathcal{L}}}{2\mathcal{L}}\sum_{j=1}^{E}
          \sum_{i=1}^{\lceil l^{\delta}_{\min}/A_{\mathcal{L}}\rceil} C_{ji}
   &\geq \frac{\epsilon_{\mathcal{L}}A^3_{\mathcal{L}}}{2\sqrt{2}E^3l^{3\delta}_{\min}}
      \frac{1}{\mathcal{L}}\omega^2_{\beta}(a^\ast(\Phi_{1})a(\Phi_{1}))-E\rho \\
   &\geq \frac{\epsilon_{\mathcal{L}}A^3_{\mathcal{L}}}{2\sqrt{2}E^{3}\mathcal{L}^{3\delta-1}}
      \Bigl(\frac{\omega_{\beta}(a^\ast(\Phi_{1})a(\Phi_{1}))}{\mathcal{L}}\Bigr)^2
      -E\rho\ .
\end{split}
\end{equation}
Defining
\begin{equation}\begin{split}
 D_{ji}:=\int_{\tilde{I}_i}\int_{\tilde{I}_i}
 \omega_{\beta}(a_j^\ast(x)a_j^\ast(y)a_j(x)a_j(y))\ \ud y\,\ud x\ ,
\end{split}
\end{equation}
one obtains in a similar way, 
\begin{equation}
\label{Dsum}
\begin{split}
 \frac{\epsilon_{\mathcal{L}}}{2\mathcal{L}}\sum_{j=1}^{E}
 \sum_{i=1}^{\lceil l^{\delta}_{\min}/A_{\mathcal{L}}\rceil} D_{ji}
   &\geq \frac{\epsilon_{\mathcal{L}}A^3_{\mathcal{L}}}{2\sqrt{2}E^{3}\mathcal{L}^{3\delta-1}}
     \Bigl(\frac{\omega_{\beta}(a^\ast(\Phi_{2})a(\Phi_{2}))}{\mathcal{L}}\Bigr)^2-
     E\rho\ .
\end{split}
\end{equation}
The right-hand sides of \eqref{Csum} and \eqref{Dsum}, therefore, provide
lower bounds to \eqref{1stest}. We choose $A_{\mathcal{L}},\epsilon_{\mathcal{L}}$ both 
either constant, or such that 
$\epsilon_{\mathcal{L}}A_{\mathcal{L}}^3=O(\mathcal{L}^{3\delta+\gamma-1})$, where 
$0<\gamma < 1-3\delta$. The latter choice is possible as $\delta<\tfrac{1}{3}$. 
Hence, the lower bounds in \eqref{Csum} and \eqref{Dsum} tend to infinity in the 
thermodynamic limit, unless \eqref{partdense} is fulfilled. 
Lemma~\ref{LemmaEnergy}, however, requires the energy density to remain finite, 
hence \eqref{partdense} follows.
\end{proof}
To prove the absence of condensation into the one-particle ground state 
we need the following statement.
\begin{lemma}
\label{PropositionCoefficients} 
Let $\Phi_0=(\varphi_e)_{e\in\mathcal{E}}\in\cH_1$ be the normalised one-particle 
ground state with components $\varphi_e(x)=a_e e^{-\sqrt{|E_0|}x}+b_e e^{+\sqrt{|E_0|}x}$ 
and corresponding eigenvalue $E_0<0$. Then,
\begin{equation}
\label{ReferenceProof1}
 \max_{e \in \mathcal{E}}\sup_{\mathcal{L}}\bigl(|\varphi_e(0)|+|\varphi_e(l_e)|\bigr) 
 < \infty\ ,
\end{equation} 
and the coefficients are such that $|a_e|=O(1)$ and $|b_e|=O(e^{-\sqrt{|E_0|}l_e})$.
\end{lemma}
\begin{proof}
The squared norm of the function $\Phi_0=(\varphi_e)_{e\in\mathcal{E}}$ is
\begin{equation}
 \|\Phi_0\|^2 = \sum_e\biggl(\frac{|a_e|^2}{2\sqrt{|E_0|}}\left( 1-\ue^{-2\sqrt{|E_0|}l_e}
 \right)+\frac{|b_e|^2}{2\sqrt{|E_0|}}\left( \ue^{2\sqrt{|E_0|}l_e}-1 \right)
 +2|\re(\bar a_eb_e)|l_e\biggr) \ .
\end{equation}
In order for this to equal one, as $l_e\to\infty$, one has to require that 
$|a_e|=O(1)$ and $|b_e|=O(e^{-\sqrt{|E_0|}l_e})$. Since $\varphi_e(0)=a_e+b_e$ and 
$\varphi_e(l_e)=a_e\ue^{-\sqrt{|E_0|}l_e}+b_e\ue^{\sqrt{|E_0|}l_e}$ the property 
\eqref{ReferenceProof1} follows.
\end{proof}
As a consequence, the one-particle ground state is localised around the
vertices of the graph. This is similar to the model in \cite{PS86} and differs
essentially from the model in \cite{LVZ03}.

We can now formulate the main result of this Note. 
\begin{theorem} 
Let $\Phi_0 \in \cH_1$ be the ground state of the one-particle system. 
Furthermore, let $H_N$ be given with interaction potential $U_{\mathcal{L}}$ as 
in Lemma~\ref{LemmaCutoff}. Then,
\begin{equation}
 \limsup_{\mathcal{L} \rightarrow \infty}\frac{\omega_{\beta}(a^\ast(\Phi_0)a(\Phi_0))}
 {\mathcal{L}}=0\ .
\end{equation} 
Hence, in the interacting many-particle system there is no condensation into the 
one-particle ground state. 
\end{theorem}
\begin{proof} 
We use the cut-offs introduced in Lemma~\ref{LemmaCutoff} and write 
$\Phi_0=\Phi_1+\Phi_2+\Phi_3$ where 
$(\Phi_3)_e:=\varphi_{0,e}\chi_{[l^{\delta}_{\min},l_e-l^{\delta}_{\min}]}$. This gives
\begin{equation}
 \frac{\omega_{\beta}(a^\ast(\Phi_0)a(\Phi_0))}{\mathcal{L}}
 =\sum_{i,j=1}^3\frac{\omega_{\beta}(a^\ast(\Phi_i)a(\Phi_j))}{\mathcal{L}}\ .
\end{equation} 
For the diagonal terms, first Lemma~\ref{LemmaCutoff} implies that
\begin{equation}
 \limsup_{\mathcal{L} \rightarrow \infty}\frac{\omega_{\beta}(a^\ast(\Phi_i)a(\Phi_i))}
 {\mathcal{L}}=0\ , \quad \quad i=1,2\ . 
\end{equation}
Then, Lemma~\ref{PropositionCoefficients} yields that 
$\lim_{\mathcal{L}\to\infty}\|\Phi_3\|_{\cH_1}=0$. Using the Cauchy-Schwarz inequality we 
hence obtain
\begin{equation}
 \limsup_{\mathcal{L} \rightarrow \infty}\frac{\omega_{\beta}(a^\ast(\Phi_3)a(\Phi_3))}
 {\mathcal{L}} \leq \limsup_{\mathcal{L} \rightarrow \infty} \frac{N}{\mathcal{L}} 
 \|\Phi_3\|^2_{\cH_1}=0\ .
\end{equation}
Using the Cauchy-Schwarz inequality the off-diagonal terms can be bounded by
the diagonal terms,
\begin{equation}
 \frac{\omega_{\beta}(a^\ast(\Phi_i)a(\Phi_j))}{\mathcal{L}} \leq 
 \frac{\omega_{\beta}(a^\ast(\Phi_i)a(\Phi_i))}{\mathcal{L}}+
 \frac{\omega_{\beta}(a^\ast(\Phi_j)a(\Phi_j))}{\mathcal{L}}\ ,
\end{equation} 
which concludes the proof.
\end{proof}

\vspace*{0.5cm}

\subsection*{Acknowledgement}
This work was supported by the EPSRC network {\it Analysis on Graphs}
(EP/1038217/1).

\vspace*{0.5cm}

{\small
\bibliographystyle{amsalpha}
\bibliography{Literature}}

\end{document}